\documentclass[a4paper,aps,pra,showpacs,superscriptaddress,nofootinbib,reprint]{revtex4-1}
\usepackage[utf8]{inputenc}
\usepackage[T1]{fontenc}
\usepackage[UKenglish]{babel}
\usepackage{lmodern}
\usepackage[colorlinks]{hyperref}
\usepackage{graphicx}
\usepackage{amsmath,amssymb,amsthm}
\usepackage{xspace}
\usepackage{mathtools}
\usepackage{dsfont}

\usepackage{paralist}
\usepackage{array}

\DeclareMathOperator{\tr}{tr}

\DeclareMathOperator{\pr}{Pr}
\DeclareMathOperator{\doo}{do}

\newcommand{\id}{\mathds{1}}

\newcommand{\ket}[1]{\left| #1 \right\rangle}

\newcommand{\ketbra}[2]{\left|#1\middle\rangle\middle\langle#2\right|}
\newcommand{\proj}[1]{[#1]}

\newcommand{\Bra}[1]{{ \langle \! \langle{#1}\vert }}
\newcommand{\Ket}[1]{{ \vert {#1}  \rangle \!  \rangle}}
\newcommand{\KetBra}[2]{{\Ket{#1}\!\Bra{#2} }}

\newcommand{\Proj}[1]{[[#1]]}

\newcommand{\overbar}[1]{\mkern 1.5mu\overline{\mkern-2.5mu#1\mkern-1.5mu}\mkern 1.5mu}

\newtheorem{theorem}{Theorem}
\newtheorem*{theorem*}{Theorem}

\newtheorem{definition}{Definition}
\newtheorem{corollary}[theorem]{Corollary}

\newcommand{\lin}{{\cal L}}

\newcommand{\ps}{P\!S}
\newcommand{\psv}{ps}
\newcommand{\opa}{S}
\newcommand{\ech}{R}

\newcommand{\out}{\textsc{\footnotesize out}\xspace}
\newcommand{\IN}{\textsc{\footnotesize in}\xspace}

\newcommand{\lab}{${\cal L}= \left\{L_j=I_j\otimes O_j\right\}_{j=1}^n$\xspace}

\newcommand{\qdag}{MQCM\xspace}

\newcommand{\qmc}{$\left\langle G, W\right\rangle$\xspace}

\newcommand{\blue}[1]{\textcolor{blue}{#1}}

\begin{document}

\title{Quantum causal modelling}
\author{Fabio Costa}
\email{f.costa@uq.edu.au}
\affiliation{Centre for Engineered Quantum Systems, School of Mathematics and Physics, The University of Queensland, St Lucia, QLD 4072, Australia}
\author{Sally Shrapnel}
\email{s.shrapnel@uq.edu.au}
\affiliation{School of Historical and Philosophical Inquiry, The University of Queensland, St Lucia, QLD 4072, Australia}
\affiliation{School of Mathematics and Physics, The University of Queensland, St Lucia, QLD 4072, Australia}

\date{\today}
\begin{abstract}
Causal modelling provides a powerful set of tools for identifying causal structure from observed correlations.  It is well known that such techniques fail for quantum systems, unless one introduces `spooky' hidden mechanisms. Whether one can produce a genuinely quantum framework in order to discover causal structure remains an open question. Here we introduce a new framework for quantum causal modelling that allows for the discovery of causal structure. We define quantum analogues for core features of classical causal modelling techniques, including the causal Markov condition and faithfulness. Based on the process matrix formalism, this framework naturally extends to generalised structures with indefinite causal order.
\end{abstract}
\maketitle

\section{Introduction}
\vspace{-0.19cm}
\subsection{Why cause}
\vspace{-0.19cm}
Causality often appears as an axiom in the formulation of physical theories, typically as the assumption that effects cannot precede their causes. However, what is often left unclear is which general principles should be used to \emph{define} causal relations and relata. 

The recent technical developments of causal modelling provide a consistent mathematical formalism that relates causal influence to the possibility of signalling \cite{spirtes2000causation, pearlbook}. This approach provides an explanation for why causation is not `merely correlation' and has shed light on contentious issues across a variety of scientific fields \cite{woodward2003making, sloman2009causal, Koller+Friedman:09}.

Causal models rest on essentially two concepts: the existence of autonomous \emph{mechanisms}, responsible for producing the observed correlations, and the possibility of \emph{interventions}, which modify part of the mechanism according to some factor external to the model.  It is the latter notion that differentiates causal relations from simple correlations. Together, mechanisms and interventions define what is known as the \emph{causal Markov condition}, a constraint imposed on the structure of a causal model.

Causal models have proved tremendously useful for understanding causal relations across a broad range of disciplines. They provide a powerful mathematical toolbox for efficiently extracting causal information from observed data, Figs.~\ref{blackbox}, \ref{watermill}.  They enable scientists to predict, manipulate and explain. The application of such models to physics represents a promising direction for a deeper and unified understanding of causality \cite{spekkens2015}. 

\begin{figure}%
\includegraphics[width=0.63\columnwidth]{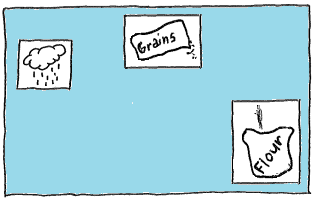}%
\caption{\textbf{Correlations.} A typical task in science is to provide a model that can explain the correlations between some observed variables.}%
\label{blackbox}%
\end{figure}

Unfortunately, one cannot straightforwardly apply these methods to quantum systems. Causal models presume the existence of objective properties that can be observed and manipulated locally. Such assumptions are incompatible with quantum mechanics, as most prominently demonstrated by Bell's theorem  \cite{bell64}. Of course, one \emph{can} apply the classical methodology, but only at the expense of introducing hidden, fine tuned \cite{woodlesson2012} mechanisms, such as action at a distance \cite{Naeger2015} or retrocausality \cite{Evans01062013}. Although seemingly conservative, such approaches violate the very pillars of causal reasoning: they posit the existence of mechanisms that cannot be discovered and variables that cannot be controlled. 

An alternative approach is to reformulate causal models from the ground up, in a way that makes direct use of the quantum formalism. Indeed, quantum theory can be seen as a theory of interventions \cite{Peres1999kz} and, from the perspective of quantum information, causal relations are identified with \emph{signalling}, in direct agreement with the intuitions underlying classical models. Additionally, quantum circuits \cite{deutsch1989quantum} are often interpreted as being representative of causal structure. This suggests that quantum causal models should be rather natural within quantum theory.
However, despite much work in this direction \cite{tucciquantum1995, Leifer2006, Laskey2007, Leifer2013, cavalcanti2014modifications,
fritzbeyond2015, Henson2015, pienaar2014graph, chavesinformation2015, ried2015quantum}, a general framework for quantum causal models has not been found. As identified in Ref.~\cite{chavesinformation2015}, a key missing step is the formulation of a quantum version of the causal Markov condition. In classical models, this condition enables one to use observational and interventionist inference to deduce causal structure from data. Arguably, a formalism for quantum causal modelling would require a similar condition.
\begin{figure}%
\includegraphics[width=\columnwidth]{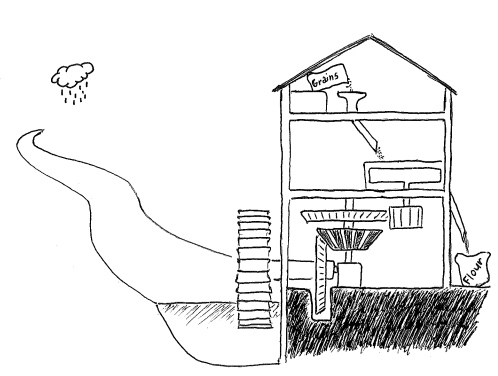}%
\caption{\textbf{Causal model}. A causal model produces correlations in terms of a mechanism. The rain fills the stream that moves the wheel, activating the mechanism that grinds the grains to produce flour.}%
\label{watermill}%
\end{figure}
\subsection{This work}
Here a complete framework for quantum causal modelling is presented. The framework rests on quantum definitions of mechanisms and spatio-temporally localised interventions. A quantum version of the causal Markov condition is defined, allowing for the possibility of \emph{causal discovery}, namely, the identification of causal structure from empirical data. Classical causal models are recovered as a limiting case of the formalism. Furthermore, the formalism allows for natural quantum extensions of classical concepts, such as faithfulness and the distinction between direct and indirect causes.

In the present framework, causal relations are identified with the possibility of signalling, that is, with the possibility of modifying the statistics of an event through intervention on a different event. 
The role of interventions was not clearly identified in previous works, which rather attempted to generalise various aspects of `causally neutral' Bayesian networks~\cite{tucciquantum1995, Leifer2006, Leifer2013, cavalcanti2014modifications,
fritzbeyond2015, Henson2015, pienaar2014graph, chavesinformation2015}. Without a definition of intervention, 
it is not clear in what sense such generalisations should be thought of as \emph{causal}.

Although some notion of interventions is present in Ref.~\cite{Laskey2007},
in that formalism not all correlations have a causal explanation, as a new form of `contemporaneous' correlation is postulated to account for entanglement. By contrast, in the causal models introduced here all correlations arise from direct, indirect, or common causes.

Classical causal models use probability distributions and it is natural to ask whether quantum causal models can be defined by simply replacing such distributions with density matrices. However, a classical joint probability distribution can describe common-cause as well as direct-cause relations, whilst a joint density matrix cannot describe direct-cause relations between quantum events \cite{Leifer2013}.

Attempts have been made to generalise density matrices in order to solve this problem~\cite{Leifer2013, fitzsimonsquantum2015}. In such attempts each quantum event is associated with a single Hilbert space, and it has been shown that such approaches face serious problems \cite{horsmanstateintime2016}.
In contrast, the present work identifies a quantum event with a local operation and thus is associated with a \emph{pair} of Hilbert spaces, describing respectively the input and the output of the operation. This follows the quantum networks~\cite{chiribella09b} and process matrix~\cite{oreshkov12} formalisms.
Ultimately, `splitting' the Hilbert space of an event is what allows for a unified representation of arbitrary causal structures. This fact was also recognised in Ref.~\cite{ried2015quantum}, which presented a special case of the general framework introduced here.

Other recent works have focussed on the correlations between classical variables produced in experiments involving quantum systems~\cite{fritzbeyond2015, Henson2015, pienaar2014graph}. These works adopt a \emph{device independent} perspective, where the observed data contain no information about the functioning of or the physical theory describing the instruments used. 
Such frameworks cannot distinguish direct from indirect causes, or more generally formulate a quantum version of the causal Markov condition. 
Such limitations are overcome here by adopting a device dependent quantum formalism, where the description of events contains all relevant information regarding the physical devices that produced them.

The work is organised as follows. In Sec.~\ref{classical} the framework of classical causal models is briefly reviewed. The quantum formalism is introduced in Sec.~\ref{quantum}, where \emph{Markov  quantum causal models} are defined and it is shown that non-Markovian, causally-ordered models can always be reduced to Markovian ones. In sec.~\ref{discovery}, faithful causal models are defined and the basic tools for \emph{quantum casual discovery} introduced. It is shown in Sec.~\ref{limit} that classical causal models are recovered as a limiting case of quantum ones. Possible extensions of the formalism to indefinite and cyclic causal structures are briefly discussed in Sec.~\ref{beyond}.
An introduction to the framework presented here is also considered from a broader philosophical perspective in Ref.~\cite{pittphilsci11751}.

\section{Classical causal models}
The brief review of classical causal models presented here is based on the monograph by Pearl~\cite{pearlbook}.
\label{classical}

\subsection{Mechanisms}

In classical physics, an event is identified with the value assumed by some property of a system at a given moment in time (denoted $x, y\dots$). \emph{Potential} events are described by random variables ($X, Y,\dots$), with values distributed according to some probability $\pr(X=x,\,Y=y,\dots)$. The short-hand $P(x)\equiv\pr(X=x)$ is used when there is no risk of confusion.

The building blocks of a causal model are mechanisms that connect events to each other according to classical laws of dynamics, see Fig.~\ref{figmechanism}. 
Deterministic laws imply a functional dependence between variables:
\begin{equation}
Y = f(X).
\label{deterministic}
\end{equation}
\begin{figure}
\includegraphics[width=0.55\columnwidth]{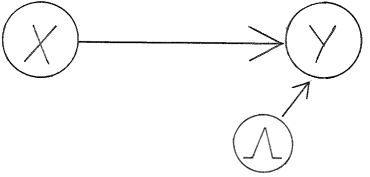}
\vspace{-1em}
\caption{\textbf{Classical Mechanism.} A classical mechanism is defined by a function $Y=f(X)$, relating random variables. The variables are depicted as circles and the mechanisms depicted via connecting arrows. External noise is represented by an additional variable $\Lambda$ with an associated probability distribution. Noise variables are typically not drawn explicitly and the arrows are interpreted as generic stochastic processes, defined by a conditional probability $P(Y|X)$.}\label{figmechanism}
\end{figure}
Unknown noise affecting the mechanism can be modelled as an additional variable $\Lambda$ with an associated probability distribution $P(\lambda)$. The functional relation is then $Y=f(X,\Lambda)$, which defines a stochastic mechanism that transforms $X$ to $Y$ according to the conditional probability
\begin{equation}
\begin{split}
P(y|x)= &\sum_{\lambda} P(y|x,\,\lambda)P(\lambda),\\
 P(y|x,\,\lambda) = &\,\delta_{y\,f(x,\,\lambda)},
\end{split}
\label{conditional}
\end{equation}
where $\delta_{a b}$ is the Kronecker delta.

\subsection{Interventions}

What distinguishes a causal mechanism from a generic conditional probability is the possibility of changing it via \emph{intervention}, Fig.~\ref{intervention}.
Causal relations are indeed defined by the possibility of intervening on a variable in the model and forcing it to assume some chosen value, thereby overriding the natural process that generates it. If $X$ is a cause of $Y$, then setting $X$ to a particular value will not change the conditional probability $P(y|x)$. On the other hand, setting $Y$ to a specific value requires breaking the process that generates $Y$, thus removing the correlations between $X$ and $Y$.
\begin{figure}
\includegraphics[width=0.55\columnwidth]{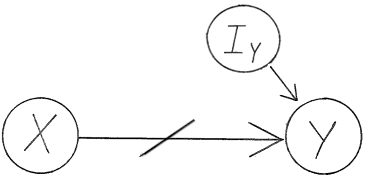}
\vspace{-1em}
\caption{\textbf{Interventions.} A \emph{causal} mechanism can be altered through an external intervention. 
An intervention can be represented as an additional random variable $I_Y$, which is commonly not drawn explicitly.}
\label{mechanism}
\end{figure}

More generally, an intervention on a variable $Y$ can be any modification of the mechanism that produces $Y$, without necessarily breaking all pre-existing causal influences. This can be modelled by some control variable $I_Y$ that parametrizes the mechanism producing $Y$. Intervention variables can be treated as random variables themselves, thus causal mechanisms are defined as conditional probabilities
\begin{equation}
\pr(Y=y|X=x,\,I_Y=i_Y).
\label{causalmech}
\end{equation}
If the only type of intervention considered is ``set $Y$ to value $y$'', the intervention variable takes as possible values
\begin{equation}
I_Y = \textrm{idle},\quad I_Y = \doo(y),
\label{intervention}
\end{equation}
for all possible values $y$ of $Y$. The `do' intervention is formalised as
\begin{equation}
\begin{array}{l}
\pr(Y=y|X=x,\,I_Y=\textrm{idle}) \equiv P^*(y|x), \\
\pr(Y=y|X=x,\,I_Y=\doo(y')) = \delta_{y y'},
\end{array}
\label{modularmech}
\end{equation}
where $P^*$ represents the `natural' probability, i.e.\ the one observed if no intervention is made, often simply written as $P(y|x)$.\footnote{Note the small abuse of notation, i.e.\ $P(y|x)$ does not represent the marginal of $P(y,i_Y|x)$, as the notation would imply.} 

It is of no particular significance whether the intervention is performed by a human agent, a pre-programmed device, or nature itself.
Intervention variables simply represent parameters that have to be fixed independently of the variables that are included in the model itself, in order to extract predictions from it.
The possibility of separating internal variables from external parameters is the fundamental assumption that allows for a causal interpretation of the model.
Accordingly, a variable $X$ is interpreted as a \emph{cause} of another variable $Y$ if correlations exist between $I_X$ and $Y$.

\subsection{Causal models} 
A causal model captures the causal relations between a set of variables in terms of mechanisms and interventions, Fig.~5.
The qualitative structure of cause-effect relations in a model defines a \emph{causal structure}, which is represented as a directed acyclic graph.
\begin{definition}[DAG]
A\emph{ directed graph }is a pair $G=\left\langle {\cal V}, {\cal E}\right\rangle$, where ${\cal V}=\left\{V_1,\dots,V_n\right\}$ is the set of \emph{vertices} and ${\cal E}\subset {\cal V} \times {\cal V}$ is a set of ordered pairs of vertices, representing \emph{directed edges}. A \emph{directed path} is a sequence of directed edges $\left\{e_k=\left(V_k^1,V_k^2\right)\in {\cal E}\right\}_{k=1}^{m}$ with $V_k^2=V_{k+1}^1$ for $k=1,\dots,m-1$. A \emph{directed cycle} is a directed path with $V_m^2=V_1^1$. A \textbf{\emph{directed acyclic graph (DAG)}} is a directed graph that contains no directed cycles.
\end{definition}

A functional causal model is defined by a set of functions that determine the observed variables given their direct causes and possibly some noise. The Markov assumption states that all noise variables are uncorrelated, leading to the following definition:
\begin{figure}
\begin{tabular}{lll}
(a) & \qquad & (b) \\
\includegraphics[width=0.39\columnwidth]{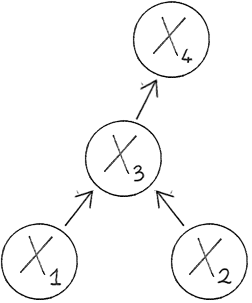} & \qquad &
\includegraphics[width=0.48\columnwidth]{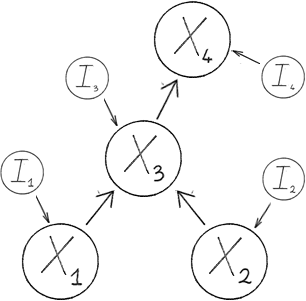}
\end{tabular}
\label{causalmodels}
\vspace{-1em}
\caption{\textbf{Classical causal models.} (a) The \emph{causal structure} of a classical causal model is represented by a \emph{directed acyclic graph}, where the nodes represent random variables and the edges classical mechanisms. A joint probability compatible with a causal structure satisfies the \emph{Markov condition}, Eq.~\eqref{markov}. (b) What makes a model \emph{causal} is the possibility, at least in principle, of performing localised interventions on each variable.}
\end{figure}
\begin{definition}[Classical causal model]
A \textbf{\emph{classical causal model}} for a set of random variables ${\cal X}=\left\{X_j\right\}_{j=1}^n$ is defined by 
\begin{enumerate}
	\item A DAG that has the random variables as vertices, $G = \left\langle {\cal X}, {\cal E}\right\rangle$;
	\item A list of conditional probabilities for each variable given its parents, $P(x_j|pa_j)$, $j=1,\dots,n$, where the set of \emph{parents} $P\!A_j$ of a vertex $X_j$ in a DAG is defined as the set of vertices $X$ with an edge from $X$ to $X_j$: 
	\begin{equation}
	P\!A_j:=\left\{X\in {\cal X}| (X,X_j)\in {\cal E}\right\}. 
	\label{parents}
	\end{equation}
\end{enumerate}
\end{definition}
The conditional probabilities $P(x_j|pa_j)$ represent autonomous causal mechanisms that can be modified through interventions. As mentioned above, the intervention variables are often left implicit, i.e.\ it is understood that $P(x_j|pa_j)\equiv \pr(X_j = x_j|P\!A_j = pa_j, I_j= \textrm{idle})$. 

A causal model generates a probability distribution 
\begin{equation}
P(x_1,\dots,x_n)=\prod_{j=1}^nP(x_j|pa_j)
\label{markov}
\end{equation}
for the observed random variables. Decomposition \eqref{markov} is called the \emph{Markov condition} and can be used to test the compatibility of given data with a causal structure, even in the absence of interventions. However, without additional assumptions, the Markov condition does not identify a unique DAG and the causal structure remains underdetermined.\footnote{In fact, even under additional assumptions such as Faithfulness and Causal Sufficiency, the causal structure will often remain underdetermined. In such situations,  interventions are required to uncover the correct structure.}  Explicitly, the causal information of a model is encoded in the conditional probability
\begin{equation}
P(x_1,\dots,x_n|i_1,\dots,i_n)=\prod_{j=1}^nP(x_j|pa_j,i_j),
\label{causalmarkov}
\end{equation}
which reduces to condition \eqref{markov} when $I_j=\textrm{idle}$ for all $j$.
Eq.~\eqref{causalmarkov} can be called \emph{causal Markov condition}, although often this is defined as condition \eqref{markov}, leaving intervention variables and causal interpretation implicit. Eq.~\eqref{causalmarkov} can also be used to define causal models in the language of \emph{influence diagrams} \cite{dawid2002}.

If a set of variables $Y_1,\dots,Y_r$ in a causal model is not observed, the marginal probability $P(x_1,\dots,x_n)=\sum_{y_1,\dots,y_r}P(x_1,\dots,x_n,y_1,\dots,y_r)$
does not necessarily satisfy the Markov condition, and the unobserved variables are called \emph{latent}, see Fig.~\ref{hidden}. Given a probability that does not fulfill the Markov condition, it is always possible to extend the set of variables to include latent variables that restore the condition.
\begin{figure}
\begin{tabular}{lll}
(a) &\; & (b) \\
\includegraphics[width=0.4\columnwidth]{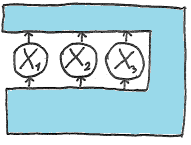} &\; &
\includegraphics[width=0.4\columnwidth]{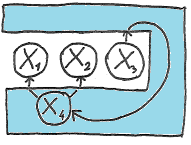}
\end{tabular}
\vspace{-1em}
\caption{\textbf{Latent variables.} (a) A probability $P(x_1,x_2,x_3)$ that does not satisfy the Markov condition can be seen as an unknown mechanism that connects the variables. In this case, it is natural to look for an extended model (b), including latent variables (here $X_4$) such that $P(x_1,x_2,x_3,x_4)$ does satisfy the Markov condition for some causal structure.}
\label{hidden}
\end{figure} 
Implicit in such a move is the assumption that  latent variables can be accessed at least in principle and their causal role verified through interventions. In the context of quantum correlations, such an extension leads to hidden variable models. However, a causal interpretation of such models implies a deviation from quantum mechanics:  intervening on hidden variables would allow signalling at a distance, in contradiction with quantum predictions. Models in which the hidden variables cannot be accessed, even in principle, are a logical possibility, but they do not have a causal interpretation in the interventionist sense that underpins classical causal modelling.

\section{Quantum causal models}
\label{quantum}
The \emph{process matrix formalism} reviewed below was introduced in Ref.~\cite{oreshkov12} as an extension of causally ordered \emph{quantum networks} \cite{gutoski06, chiribella08, chiribella09b}. Only finite-dimensional systems are considered here, a generalisation to infinite dimensions was studied in Ref.~\cite{Giacomini2015}.

\subsection{Quantum events}

In contrast to classical systems, quantum systems do not possess objective properties that can be assigned values prior to and independently of measurement. Thus, the relata of causal relations are not classical random variables---a genuinely quantum notion of event is needed. A quantum event can be identified with all available information about a system localised in space and time. Typical quantum events are ``the spin was found aligned with the $z$ axis'' and ``the spin was rotated by $\pi$ around the $z$ axis''.
A combination of the type ``the spin was found aligned with the $z$ axis and then reprepared aligned with the $x$ axis'' can be considered as an event too. In general, a quantum event is associated with an operation, which can be deterministic (as a fixed transformation) or non-deterministic (as one associated with the outcome of a measurement).

Formally, such a quantum operation is represented by a \emph{completely positive} (CP) map ${\cal M}:A_I\rightarrow A_O$, where input and output spaces are the spaces of linear operators over input and output Hilbert spaces, $A_I\equiv\lin({\cal H}^{A_I})$, $A_O\equiv \lin({\cal H}^{A_O})$, respectively (here identified with the corresponding matrix spaces).
Complete positivity means that, for arbitrary dimensions of an ancillary system $A'$, the map ${\cal I}_{A'}\otimes{\cal M}$ transforms positive operators into positive operators, where ${\cal I}_{A'}$ is the identity map on
${A'}$. Input and output spaces can have different dimensions $d_{A_I}$, $d_{A_O}$, as ancillas can be added or discarded. 

Using the Choi-Jamio{\l}kowski isomorphism \cite{jamio72, choi75}, a CP map can be represented as a matrix\footnote{This definition aligns with the convention in Ref.~\cite{oreshkov12}. Different choices, amounting to a transposition or partial transposition, are found elsewhere.}:
\begin{eqnarray}
\label{CJ}
M^{A_I A_O} =& \sum_{j\,l}\ketbra{l}{j}^{A_I}\otimes \left[{\cal M}(\ketbra{j}{l})^{A_O}\right]^T , \\ \label{inverseCJ}
{\cal M}(\rho)^{A_O} =& \left[\tr_{A_I}\left(\rho^{A_I}\otimes \id^{A_O}\cdot M^{A_IA_O}\right)\right]^T,
\end{eqnarray}
where $\left\{\ket{j}\right\}_{j=1}^{d_{A_I}}$ is an orthonormal basis in ${\cal H}^{A_I}$ and $^T$ denotes transposition in that basis\footnote{The convention used here is that the superscript denotes the space in which a matrix is defined---hence $M^{X}$ is a short-hand for $M\in X$ (and ${\cal J}^X$ for ${\cal J}\subset X$, if ${\cal J}$ is a set of matrices)---but not the matrix itself: $T^{AB}=M^A\otimes M^B$ means that $T$ is the tensor product of two equal matrices $M$. Identity matrices are sometimes omitted: $M^A\otimes\id^B\equiv M^A$, etc.}. According to Choi's theorem, ${\cal M}$ is CP if and only if $M^{A_I A_O}\geq 0$. Henceforth CP maps will be identified with their Choi-Jamio{\l}kowski representation unless otherwise stated. 

An example of a quantum event is a projective measurement that yields as an outcome some pure state $\ket{\psi}$ and resets the systems in the same state. As a CP map, this is represented as 
\begin{equation}
M^{A_IA_O}_{\psi} = \proj{\psi}^{A_I}\otimes\left(\proj{\psi}^{A_O}\right)^T,
\label{projective}
\end{equation}
where the notation $\proj{\psi}:=\ketbra{\psi}{\psi}$ is used \cite{CohenWhen11}. Another example is a transformation of the system occurring with unit probability, for example as defined by a unitary matrix $U$. Introducing the notation \cite{royer_wigner_1991, braunstein_universal_2000}
\begin{equation}
\Ket{U}:= \sum_j \ket{j}U\ket{j},\quad \Proj{U}:=\KetBra{U}{U},
\label{doubleket}
\end{equation}
the local event corresponding to the unitary transformation $U$ is given by
\begin{equation}
M^{A_IA_O}_{U} = \Proj{U^*}^{A_I A_O},
\label{localunitary}
\end{equation}
where $^*$ is the complex conjugation in the chosen basis. A more general deterministic transformation is a CP and trace-preserving (CPTP) map, represented as
\begin{equation}
M^{A_IA_O}\geq 0,\quad \tr_{A_O}M^{A_IA_O}=\id^{A_I}.
\label{localCPTP}
\end{equation}
A CPTP map is also called a \emph{quantum channel}.

The space of \emph{potential events} is identified with the set of CP maps between an input ($A_I$) and an output ($A_O$) space, which is isomorphic to (the cone of positive matrices in) the space $A_I\otimes A_O$. Input and output spaces can be respectively identified with the past and the future space-like boundaries of the space-time region where the event takes place, Fig.~\ref{laboratory}, as in Oeckl's `general boundary' formalism~\cite{Oeckl2003318}. The space of potential events, which plays a role analogous to a classical random variable, is called a \emph{local laboratory}. 
\begin{figure}%
\includegraphics[width=0.6\columnwidth]{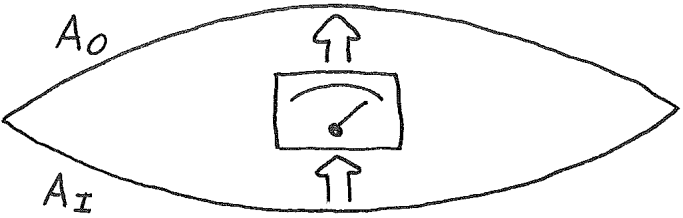}%
\caption{\textbf{Local laboratory.} A local laboratory represents a quantum system in a region of space-time with space-like boundaries. Past and future boundaries are identified with an input and an output state space, $A_I$ and $A_O$, respectively. The local laboratory is identified with the product space $A_I\otimes A_O$. A \emph{quantum event} is a quantum operation that takes place in the local laboratory and is represented as a completely positive map from input to output space.}%
\label{laboratory}%
\end{figure}

Whilst the name implies some anthropocentric element, the term `local laboratory' is taken to represent an observer-independent spatio-temporal region where quantum events can take place. In direct analogy to the classical case \cite{pricecausation13}, one can argue whether it is possible to entirely remove the concept of observer, though we shall not enter into such discussions here. 
 
\subsection{Mechanisms}
As in the classical case, a causal model is built on the notion of physical mechanisms that are responsible for mediating causal influences between events. A quantum mechanism maps the output space of a local laboratory, say $A_O$, to the input space of another one, say $B_I$. The analogue of a deterministic mechanism is a unitary\footnote{In fact, the direct analogue of a deterministic mechanism is an isometry, for which $V^{\dag}V=\id$ but it is not necessarily true that $VV^{\dag}=\id$. Unitaries represent reversible deterministic mechanisms.} map $U$, which transforms states as $\rho \mapsto U\rho U^{\dag}$. External noise can be described by an interaction with an environment which is then traced out, leading to the more general definition of a mechanism as a CPTP map.

 In the present approach, CPTP maps representing events in local laboratories are distinguished from CPTP maps representing connecting mechanisms. This distinction is reflected by a different representation of the two: a CPTP map $\cal T$ corresponding to a connecting mechanism is described by the transpose of representation \eqref{CJ}:
\begin{eqnarray} 
\label{CJt}
T^{A_O B_I} =& \sum_{j\,l}\ketbra{j}{l}^{A_I}\otimes {\cal T}(\ketbra{j}{l})^{A_O} , \\ \label{inverseCJt}
{\cal T}(\rho)^{B_I}=& \tr_{A_O}\left[\rho^{A_O}\cdot T^{A_O^T B_I}\right],\\ 
T\geq 0,& \tr_{B_I} T^{A_O B_I} = \id^{A_O},
\label{CPTP}
\end{eqnarray}
where $T^{X^TY}$ denotes partial transposition on subsystem $X$. 

As an easy example, a unitary transformation from the output space of $A$ to the input space of $B$ is represented by the projector $\Proj{U}^{A_O B_I}$,  the transpose of \eqref{localunitary}. A mechanism that connects a space of dimension $1$ to the input of a local laboratory
is given by a density matrix $\rho^{A_I}\geq 0$, $\tr \rho=1$, which corresponds to the situation where state $\rho$ is sent to the input space $A_I$.

\subsection{Interventions}

Recall, the possibility of intervention is required to give causal models empirical meaning. Quantum interventions can be formalised using the notion of \emph{instrument} \cite{davies70}, which represents the collection of all possible events that can be observed given a specific choice of probing the system. An intervention is thus defined as a \emph{choice of instrument}. 
Given a local laboratory $A_I\otimes A_O$, an instrument is formally defined as a set of CP maps that sum up to a CPTP map:
\begin{equation}
\begin{split}
{\cal J} = \left\{M_{x}^{A_IA_O}\right\}_{x=1}^m&, \quad M_{x}^{A_IA_O}\geq 0,  \\
\tr_{A_O}\sum_{a=1}^mM_{x}^{A_IA_O}=\id^{A_I}&.
\end{split}
\label{instrument}
\end{equation}
(The trace-preserving condition for the sum guarantees that probabilities sum up to $1$.)

A typical example of an instrument is the measurement of the incoming system according to a positive operator valued measure (POVM) $\left\{E_x\right\}_{x=1}^m$, $E_x^{A_I}\geq 0$, $\sum_{x=1}^mE_x^{A_I}=\id^{A_I}$, followed by the preparation of a state $\rho$. Using representation \eqref{CJ}, this is given by
\begin{equation}
{\cal J} = \left\{E_{x}^{A_I}\otimes(\rho^{A_O})^T\right\}_{x=1}^m.
\label{measureprepare}
\end{equation}
By keeping the POVM fixed and choosing different states, one breaks the flow of information through the local laboratory and obtains the equivalent of classical `do' interventions.
As discussed in Sec.~\ref{limit} below, classical 'idle' interventions correspond to quantum projective measurements in a fixed basis. These, however, are not interpreted as passive observations in the quantum formalism. Indeed, it is one of the essential features of quantum mechanics that measurements necessarily disturb the system.

\subsection{Process matrices and causal relations}
\label{processmatrix}
The general situation of interest is described by a set of local laboratories \lab,
interpreted as representing $n$ disjoint space-time regions, each bounded by two space-like surfaces. In an individual run of an experiment, instruments ${\cal J}^{L_1}_{1},\dots,{\cal J}^{L_n}_{n}$ are implemented in these local laboratories and the corresponding events recorded. The events are described by CP maps $M^{L_1}_{1},\dots,M^{L_n}_{n}$. 

By assuming the local validity of quantum mechanics, with no further assumption about causal relations between local laboratories, it is possible to prove that the probability for such events to occur is given by the generalised Born rule \cite{oreshkov12}
\begin{multline}
P(M^{L_1}_{1},\dots,M^{L_n}_{n}|{\cal J}^{L_1}_{1},\dots,{\cal J}^{L_n}_{n}) \\ =\tr\left[M^{L_1}_{1}\otimes\dots\otimes M^{L_n}_{n} \cdot W^{L_1\dots L_n}\right],
\label{born}
\end{multline}
where $W^{L_1\dots L_n}\geq 0$ is called the \emph{process matrix}, Fig.~\ref{processmat}, and represents the information about the outside world available in the local laboratories\footnote{To be more precise, the conditional probability is defined as Eq.~\eqref{born} when all the maps belong to the corresponding instruments, $M_j\in{\cal J}_j$ for $j=1,\dots,n$, and it is 0 otherwise.}.
\begin{figure}%
\includegraphics[width=0.6\columnwidth]{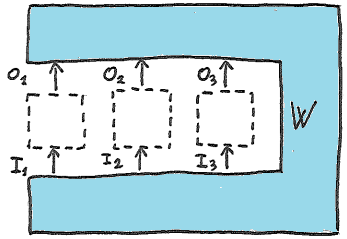}%
\caption{\textbf{Process matrix.} Given a set of local laboratories $I_1\otimes O_1, I_2\otimes O_2,\dots$ (represented as dashed boxes), all information about the ‘outside world’ is encoded in a joint probability distribution for the local events.
The probability can be calculated from the \emph{process matrix} $W^{I_1O_1I_2O_2\dots}$, defined on the tensor product of all input and output spaces.}%
\label{processmat}%
\end{figure}
In particular, the process matrix encodes all the information about the causal relations between laboratories.

 As in the classical case, causal relations are defined by signalling. By definition, a laboratory $L_h$ can signal to a laboratory $L_k$ if there exists a set of instruments $\left\{{\cal J}^{L_j}_{j}\right\}_{j=1}^n$ such that
\begin{multline}
P(M_{k}^{L_k}|{\cal J}_{1}^{L_1},\dots,{\cal J}_{h}^{L_h},\dots,{\cal J}_{n}^{L_n})\\
 \not=P(M_{k}^{L_k}|{\cal J}_{1}^{L_1},\dots,\widetilde{{\cal J}}_h^{L_h},\dots,{\cal J}_{n}^{L_n})
\label{signalling}
\end{multline}
for some instrument $\widetilde{{\cal J}_h}^{L_h}$, where the marginal probability is defined as
\begin{multline}
P(M_{k}^{L_k}|{\cal J}_{1}^{L_1},\dots,{\cal J}_{n}^{L_n})\\
:= \hspace{-4pt} \sum_{\substack{M_{r}\in {\cal J}_{r}\\ r\not=k}} \hspace{-4pt} \tr\left[M^{L_1}_{1}\otimes\dots\otimes M^{L_n}_{n}\cdot W^{L_1\dots L_n}\right].
\label{marginal}
\end{multline}
 Signalling between sets of laboratories is similarly defined. By definition, a local laboratory $A$ represents a \emph{cause} for a distinct laboratory $B$ if $A$ can signal to $B$. 

In classical causal models there is an important distinction between direct and indirect causes. Intuitively, a direct cause always influences its effect, no matter what else is changed in the model. This is formalised as follows:
\begin{definition}[direct cause] \label{direct}
Given a set of local laboratories \lab and a process matrix $W^{L_1\dots L_n}$, a laboratory $L_h$ represents a \textbf{direct cause} for a laboratory $L_k\not=L_h$ if, for \emph{any possible} set of instruments $\left\{{\cal J}^{L_j}_{j}\right\}_{j\not=k}$, there exist instruments ${\cal J}^{L_k}_{k}$ and $\widetilde{{\cal J}_h}^{L_h}$ such that
\begin{multline}
P(M_{k}^{L_k}|{\cal J}_{1}^{L_1},\dots,{\cal J}^{L_k}_{k},\dots,{\cal J}_{h}^{L_h},\dots,{\cal J}_{n}^{L_n})\\
 \not=P(M_{k}^{L_k}|{\cal J}_{1}^{L_1},\dots,{\cal J}^{L_k}_{k},\dots,\widetilde{{\cal J}_h}^{L_h},\dots,{\cal J}_{n}^{L_n}).
\end{multline}
\end{definition}

\subsection{Examples}
Before defining quantum causal models in full generality, it is useful to go through some explicit examples.

\subsubsection{Single laboratory}

Consider first the case of a single laboratory $A$, with only non-trivial input space $A_I$. In this case, process matrices reduce to $W^{A_I}\equiv\rho^{A_I}\geq 0$, CP maps reduce to POVM elements $E^{A_I}\geq 0$, and the generalised Born Rule \eqref{born} reduces to the ordinary Born rule: $P(E^{A_I})=\tr E \rho$. This describes the situation where laboratory $A$ receives a quantum state $\rho$ from the outside environment and performs a measurement on it. The environment, responsible for preparing state $\rho$, describes some part of the world that is not under experimental control. Of course it is always possible to consider a different scenario in which the state preparation is controlled. This is formalised by introducing a second laboratory $B$, with only output space $B_O$, in which the state is prepared. Now $B$ has the possibility to choose from different instruments, i.e.\ prepare different states $\rho_i$. Recall that, as local events, such state preparations are represented as $\left(\rho_i^{A_I}\right)^T$. The process matrix now describes how the two laboratories are connected: $\widetilde{W}^{B_OA_I}= \Proj{\id}^{B_O A_I}$ if $B$ is connected to $A$ via the identity channel, where notation \eqref{doubleket} is used,
\begin{equation}
\Proj{\id}^{B_OA_I}= \sum_{l\,k} \ketbra{l}{k}^{B_O}\otimes \ketbra{l}{k}^{A_I}.
\label{identitychannel}
\end{equation}
The generalised Born rule \eqref{born} now reads
\begin{multline}
P\left(E^{A_I} \middle|(\rho_i^{B_O})^T\right) \\ = \tr\left[ E^{A_I}\otimes \left(\rho_i^{B_O}\right)^T\cdot\Proj{\id}^{B_OA_I}\right] = \tr E \rho_i.
\label{examplestate}
\end{multline}
This reduces to the previous single-laboratory case if $B$ prepares $\rho_i\equiv \rho$. As the probability for any $E\not =\id$ depends non-trivially on $\rho$, $B$ represents a cause for $A$ in this example.

\subsubsection{Common cause}
The next example consists of two laboratories with non-trivial input and output spaces, $A=A_I\otimes A_O$ and $B=B_I\otimes B_O$.  A process matrix of the form $W^{AB}=\rho^{A_IB_I}\otimes \id^{A_OB_O}$ describes the situation where the two laboratories have no causal influence over each other (Fig.~\ref{common}). 
\begin{figure}%
\begin{tabular}{lll}
(a) & \quad & (b) \\
\includegraphics[width=0.4\columnwidth]{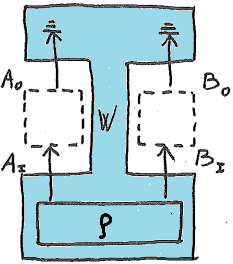} & \quad &
\includegraphics[width=0.4\columnwidth]{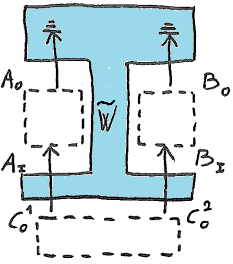}%
\end{tabular}
\caption{\textbf{Common cause.} (a) A process matrix of the form $W^{AB}=\rho^{A_IB_I}\otimes\id^{A_OB_O}$ generates non-signalling correlations between $A$ and $B$ (the `earth ground' symbol represents the identity matrix). (b) The state $\rho$ is interpreted as a common cause by introducing a third laboratory $C$. By preparing different states, $C$ can modify the correlations observed by $A$ and $B$. A new process matrix $\widetilde{W}^{ABC}$ represents the connections between the three laboratories.}%
\label{common}%
\end{figure}
The generalised Born rule reduces to
\begin{equation}
P(M^A,N^B)=\tr\left(E^{A_I} \otimes F^{B_I} \cdot \rho^{A_IB_I}\right),
\label{nosignalling}
\end{equation}
where $E^{A_I}:=\tr_{A_O} M^{A_IA_O}$ and $F^{B_I}:=\tr_{B_O} N^{B_IB_O}$.
Thus, the the above process matrix describes a bipartite state shared by the two laboratories. 

As before, one can allow for control over the state preparation by introducing a third laboratory $C$ with only non-trivial output space $C_O$. As $C$ can prepare bipartite states, its output space must decompose as a tensor product  $C_O=C_O^1\otimes C_O^2$. The subsystems $C_O^1$ and $C_O^2$ are isomorphic to $A_I$ and $B_I$ respectively. The process matrix that links $C_O^1$ to $A_I$ and $C_O^2$ to $B_I$, through identity channels, is given by $\widetilde{W}^{ABC}=\Proj{\id}^{C_O^1A_I}\otimes \Proj{\id}^{C_O^2A_I}\otimes\id^{A_OB_O}$. It is easy to see that, when the instrument ${\cal J}^C\equiv \left(\rho^{C_O^1C_O^2}\right)^T$ is chosen, the generalised Born rule reduces to expression \eqref{nosignalling}.

As a particular case, state $\rho$ can be entangled and $A$, $B$, can perform measurements that violate Bell inequalities. The \emph{common cause} for such correlations is simply associated with laboratory $C$, i.e.\ with the possibility of preparing different states. 

\subsubsection{Direct and indirect cause}
Still in the scenario with two laboratories $A=A_I\otimes A_O$,  $B=B_I\otimes B_O$, consider now the process matrix $W^{AB}= \rho^{A_I}\otimes \Proj{U}^{A_OB_I}\otimes \id^{B_O}$, where $U$ is a unitary matrix. In this case the generalised Born rule \eqref{born} reduces to
\begin{equation}
\tr\left( M^A\otimes N^B\cdot W^{AB}\right) = \tr\left[ F \; U {\cal M}(\rho)U^{\dag}\right],
\label{directAB}
\end{equation}
where $F$ is defined by $F^{B_I}:=\tr_{B_O} N^{B_IB_O}$ and ${\cal M}$ is related to $M$ through isomorphism \eqref{inverseCJ}. Eq.~\eqref{directAB} describes a state $\rho$ which is first measured in laboratory $A$ and, after evolving via unitary $U$, is measured in laboratory $B$.
For arbitrary $U$ and $\rho$, $A$ can signal to $B$ through an appropriate choice of instruments (for example preparing different pure states). Thus $A$ represents a cause for $B$.

As in the previous examples, the connecting mechanism can be included as an event in the model by introducing a third laboratory $C=C_I\otimes C_O$, with the new process matrix $\widetilde{W}^{ABC}= \rho^{A_I}\otimes \Proj{\id}^{A_OC_I}\otimes \Proj{\id}^{C_OB_I} \otimes \id^{B_O}$. The unitary evolution $U$ is now represented by the single-element instrument $\Proj{U^*}^{C_IC_O}$ (see Eq.~\eqref{localunitary}). It is now possible to consider interventions in the additional laboratory $C$. For example, $C$ can implement an instrument that breaks the flow of information from $A$ to $B$, e.g.\ the maximally noisy channel $\id^C/d_{C_O}$. In this case, no choice of instrument at $A$ can affect the probability for any POVM element in $B$. Thus, $A$ is a direct cause of $B$ given the process $W^{AB}$ but it is an indirect cause given the process $\widetilde{W}^{ABC}$. As in the classical case, whether a causal relation is direct or indirect  is relative to the set of variables included in the model.

\subsection{Markov  quantum causal models}\label{MQCM}
As for classical causal models, it is useful to depict causal relations graphically. Local laboratories are represented as nodes and the causal mechanisms connecting them as arrows.

Multiple arrows departing a single node represent different physical systems. Thus the output space of a laboratory factorises as a tensor product, with a tensor factor associated to each outgoing arrow, see Fig.~\ref{outputfactor}.
\begin{figure}%
\includegraphics[width=0.55\columnwidth]{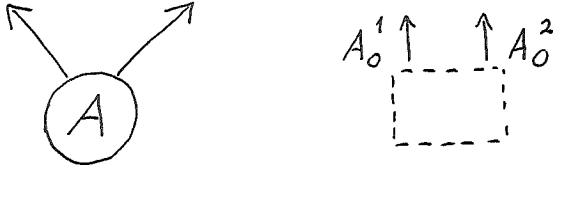}%
\caption{\textbf{Outgoing edges.} To each edge is associated a \emph{source space}. The output space of a local laboratory is the tensor product of the source spaces associated with all outgoing edges. In the picture, $A_O=A_O^{1}\otimes A_O^{2}$.}%
\label{outputfactor}%
\end{figure} 
Formally, a \emph{source space} $\opa_e$ is associated to the origin of each directed edge $e$, such that each output space factorises as $A_O=\bigotimes_{e\in \out_A}\opa_e$, where $\out_A$ is the set of edges departing from $A$.

As with classical models, arrows should represent a generic possibility of causal influence (and not just the undisturbed transfer of a system, as in the representation of quantum circuits). Thus, a graph with arrows that connect $A$ to $C$ and $B$ to $C$ represents quantum systems exiting from $A$ and $B$ and, possibly after interaction, entering $C$. In other words, the graph represents a generic quantum channel $T^{(A_O B_O)C_I}$ from the output of laboratories $A$ and $B$ to the input of laboratory $C$, Fig.~\ref{collider}
(or, more generally, from a subsystem of $A_O\otimes B_O$ to $C_I$, if $A$ and $B$ have outgoing edges that do not end on $C$).

\begin{figure}%
\includegraphics[width=0.33\columnwidth]{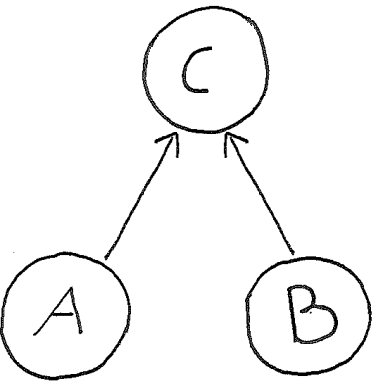}  \qquad \qquad
\includegraphics[width=0.23\columnwidth]{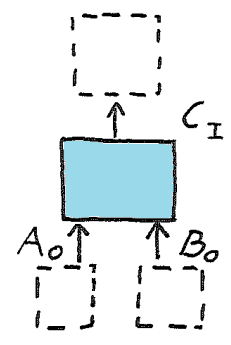}%
\caption{\textbf{Incoming edges.}
The \emph{parent space} $\ps_C$ of a local laboratory $C$ is the tensor product of the source spaces associated with all the incoming edges. (In the picture, $\ps_C = A_O\otimes B_O$.) The incoming edges represent a quantum channel from the parent space to the input space. }%
\label{collider}%
\end{figure}

When there are no outgoing edges from a laboratory, the output space can be understood as connecting to the trivial space (of dimension one). As the only CPTP map from a space $D_O$ to the trivial space is $T^{D_O}=\id^{D_O}$, a process matrix must be equal to the identity on the output space of all such laboratories.

A quantum causal model can be defined as a set of unitaries that connect (a subsystem of) the output spaces of the parents of a laboratory, plus  some noisy environment, to the input space of the laboratory (possibly after discarding part of the system).

Recall for classical causal models, the causal Markov condition is guaranteed by the assumption that all environmental noise remains uncorrelated. Also known as the independent noise assumption, this feature of causal modelling is retained in the quantum case. Thus a quantum causal model is a list of quantum channels, the structure of which forms a DAG, Fig.~\ref{qdagfig}.
\begin{definition}[\qdag]\label{qdag}
Given a set of local laboratories \lab, a \emph{\textbf{Markov quantum  causal model (\qdag)}} is a pair \qmc, where
\begin{enumerate}
	\item $G = \left\langle {\cal L},{\cal E}\right\rangle$ is a DAG that has the local laboratories as vertices;
	\item to each edge $e\in{\cal E}$ is associated a space
	$\opa_e$ such that $O_j=\bigotimes_{e\in \out_j}\opa_e$, $j=1,\dots,n$, where 
	\begin{equation}
	\out_j:=\left\{e\in{\cal E}|e=(L_j,L_k)\right\}
	\label{outgoing}
	\end{equation}
	is the set of edges departing from the vertex $L_j$;
	\item $W$ is a process matrix of the form 
	\begin{equation}
W^{L_1\dots L_n} =\bigotimes_{j=1}^n T_j^{\ps_j I_j}\otimes \id^{O_D},
\label{quantummarkov}
\end{equation}
where $O_D:=\bigotimes_{k\in{\cal D}}O_k$ is the output space of the laboratories with no outgoing edges, ${\cal D}:=\left\{k|\out_k = \varnothing \right\}$; $\ps_j:=\bigotimes_{e\in \IN_j}\opa_e$ is the \emph{parent space} associated with laboratory $L_j$, with
\begin{equation}
\IN_j:=\left\{e\in{\cal E}|e=(L_k, L_j)\right\}
\label{incoming}
\end{equation}
the set of incoming edges to $L_j$; and 
\begin{equation}
T_j^{\ps_j I_j}\geq 0,\quad \tr_{I_j}T_j^{\ps_j I_j} = \id^{\ps_j},\quad j=1,\dots,n.
\label{modelparameters}
\end{equation}
\end{enumerate}
\end{definition}
\begin{figure}%
\begin{tabular}{ll}
(a) & \qquad (b) \\
\includegraphics[width=0.37\columnwidth]{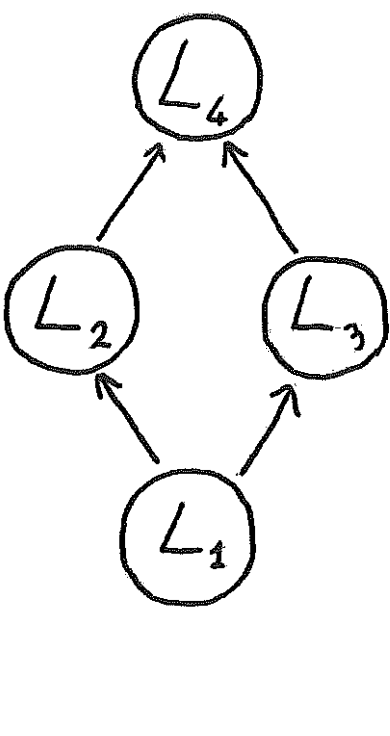} & \qquad
\includegraphics[width=0.5\columnwidth]{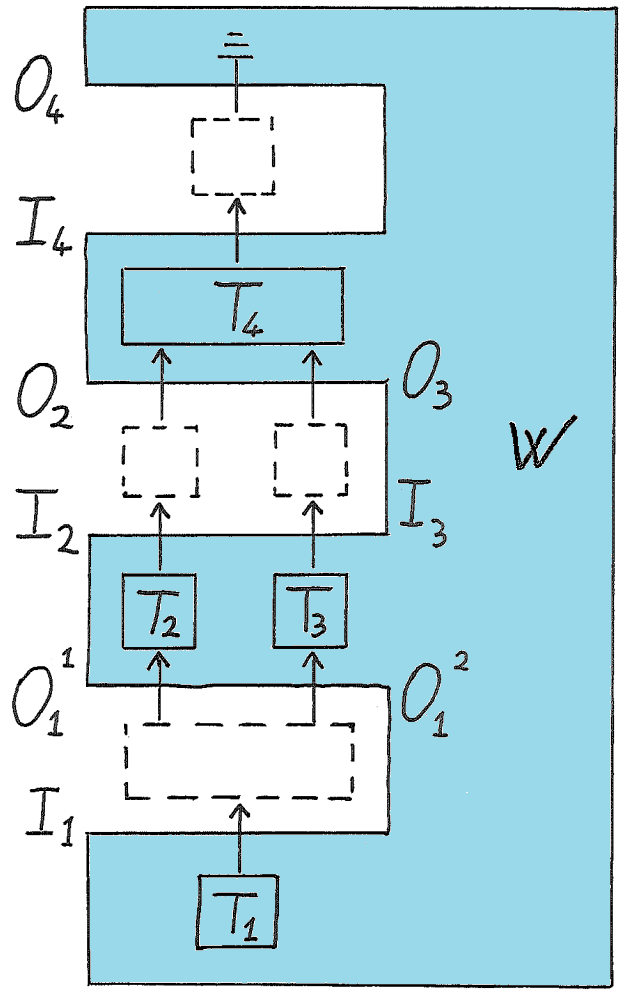} 
\end{tabular}
\caption{\textbf{Markov  quantum causal model.} (a) The causal structure of a Markov  quantum causal model is represented by a directed acyclic graph, where the nodes represent local laboratories. (b) The causal mechanisms consist of a quantum channel for each laboratory, connecting its parent space to its input space. The model is represented compactly by a process matrix that satisfies the \emph{quantum causal Markov condition}, Eq.~\eqref{quantummarkov}. For the example in the picture, $W=T_1^{I_1}\otimes T_2^{O_1^{1}I_2}\otimes T_3^{O_1^{2}I_3}\otimes T_4^{(O_2 O_3) I_4}\otimes\id^{O_4}$.}
\label{qdagfig}%
\end{figure}%
The matrices $T_j$ in Eq.~\eqref{modelparameters} define the quantum channels that connect local laboratories to each other and replace the conditional probabilities that define causal mechanisms in a classical model. 
A process matrix with the structure \eqref{quantummarkov} is said to factorise over the DAG $G$. As this factorisation is a direct analog of the causal Markov condition \eqref{causalmarkov}, it is natural to name \eqref{quantummarkov} the \emph{quantum causal Markov condition}. 

It is useful to relate Markov causal models to Markovian time evolution, a concept  familiar to most physicists \cite{breuer2002theory}. Consider a DAG $L_1\rightarrow L_2\rightarrow\dots \rightarrow L_n$, with $d_{I_j}=d_{O_j}$ (the generalisation to different dimensions is straightforward).
The quantum causal Markov condition \eqref{quantummarkov} reduces to
\begin{equation}
W=\rho_1^{I_1}\otimes T_2^{O_1I_2}\otimes\dots\otimes T_n^{O_{n-1}I_n}\otimes \id^{O_n}.
\label{markovchain}
\end{equation}
If all laboratories except $L_k$ perform the identity channel, $\Proj{\id}^{L_j}$ for $j\not =k$, and $L_k$ performs a POVM measurement, ${\cal J}^{L_k}=\left\{E_{x_k}^{I_k}\otimes \rho^{O_k}\right\}_{x_k}$, the generalised Born rule \eqref{born} reduces to
\begin{gather}\label{kmarkov}
P(E^{I_k}_{x_k})= \tr\left[E_{x_k} {\cal T}^{(k)}\left(\rho_1\right)\right],\\\label{discrete}
{\cal T}^{(k)}:= {\cal T}_k\circ\dots\circ{\cal T}_2,
\end{gather}
where $\circ$ denotes composition of maps.
Eq.~\eqref{discrete} defines a discrete-time Markovian evolution, while Eq.~\eqref{kmarkov} shows that the $k$-th laboratory receives the result of such an evolution applied to the initial state $\rho_1$, provided that no intermediate operation is performed. Channels of the form \eqref{discrete} are also called \emph{divisible }\cite{ledivisible2015}.

\subsection{Latent laboratories and non-Markovian models}\label{latent}
A quantum process can involve unobserved events. Such events may correspond to naturally occurring evolution in the environment surrounding the observed events, or to non-selective measurements, with unknown outcomes, performed by adversary agents. A local laboratory in which the events are not observed will be called \emph{latent}, in direct analogy to latent variables in classical models. 

\begin{definition}
A \emph{\textbf{Markovian causal explanation}} for a process matrix $W^{L_1\dots L_n}$ is an \qdag $\left\langle G, \widetilde{W}\right\rangle$, where $G$ is a DAG with vertices containing ${L_1,\dots, L_n}$, and possibly latent laboratories $\widetilde{\cal L} =\left\{\tilde{L}_1,\dots, \tilde{L}_m\right\}$, such that
\begin{multline}
W 
=\tr_{\widetilde{\cal L}}\left[C^{\tilde{L}_1}_1\otimes\dots\otimes C^{\tilde{L}_m}_m \cdot \widetilde{W}^{L_1\dots L_n \tilde{L}_1\dots \tilde{L}_m} \right]
\label{reduced}
\end{multline}
for some CPTP maps $C^{\tilde{L}_1}_1,\dots, C^{\tilde{L}_m}_m$, where $\tr_{\widetilde{\cal L}}$ denotes the partial trace over all the laboratories in $\widetilde{\cal L}$.
\end{definition}

In Eq.~\eqref{reduced}, W is called \emph{reduced process matrix} \cite{araujo15, oreshkov15} and provides a full description of the physical situation for the observed laboratories, once the CPTP maps in the latent laboratories are fixed.
The observed laboratories can now be linked by non-Markovian evolution, possibly with initial system-environment correlations: all the information regarding initial state and evolution of the environment is encoded in the latent laboratories. A formalism closely related to the one described here has in fact been developed for the study of non-Markovian dynamics \cite{modioperational2012, pollockcomplete2015}.

A process matrix is called \emph{causally ordered} if it is possible to define a relation of partial order among the laboratories, and signalling from a laboratory $L_1$ to another $L_2$ is possible only if $L_1$ precedes $L_2$ according to the assigned partial order. An important result from the theory of quantum networks ~\cite{gutoski06, chiribella09b} is that every causally ordered process matrix (formally equivalent to a quantum strategy or a quantum comb) can be realised by combining a sequence of quantum channels with memory. This implies that every causally ordered process matrix has a Markovian causal explanation. In other words, causally ordered, non-Markovian models can always be reduced to Markovian ones through the introduction of latent laboratories.
  
An MQCM can be further extended by including connecting mechanisms as locally observed events. There is a general procedure to re-write a connecting mechanism as a local event. Consider a local laboratory $A=A_I\otimes A_O$, with parent space $\ps_A=\bigotimes_{e\in\IN_A}\opa_e$, and a process matrix that factorises as $W=T^{\ps_A \,A_I}\otimes W'$, where $T^{\ps_A\, A_I}$ is the mechanism to be included as an event. The extended process matrix is $\widetilde{W}=\Proj{\id}^{\ps_A B_I}\otimes \Proj{\id}^{B_O A_I}\otimes W'$, with the additional local laboratory $B=B_I\otimes B_O$, $B_I\cong \ps_A$ and $B_O\cong A_I$. The original process matrix is recovered as the reduced process matrix
\begin{equation*}
W=\tr_B\left[ \left(T^{B_I B_O}\right)^T \cdot \widetilde{W}\right].
\end{equation*}

With sufficient resources, it is in principle possible to take control over all the quantum channels that define an MQCM. In this scenario, all nontrivial evolution takes place in local laboratories and can be in principle controlled, while the process matrix only describes connections between laboratories. In the resulting causal model, each edge $e$ is associated both with a \emph{source} space $\opa_e$ and a \emph{target} space $\ech_e$, such that $\opa_e\cong\ech_e$, $O_j=\bigotimes_{e\in\out_j}\opa_e$ and $I_j=\bigotimes_{e\in\IN_j}\ech_e$ for every $j$.
The process matrix is
\begin{equation}
W^{L_1\dots L_n}=\bigotimes_{e\in{\cal E}} \Proj{\id}^{\opa_{e}\ech_{e}}.
\label{connectionsonly}
\end{equation}
This process represents the set of `wires' obtained from a quantum circuit after removing all the gates.

Thus it is possible to give a causal interpretation to quantum circuits and related structures \cite{hardy2009foliable, PhysRevA.81.062348, coecke2010quantum} by allowing the possibility of replacing each gate with an arbitrary operation.

\section{Causal discovery}\label{discovery}

The success of classical causal models is largely due to the existence of efficient tools for discovering causal structures. Apart from practical utility, causal discovery has foundational significance, as it gives empirical meaning to causal structures. This section shows that  it is always possible to uniquely determine the causal structure of a Markov  quantum causal model from experimental observations. The assumptions required are  similar to those imposed on classical causal models. The question of whether this can be done efficiently will not be addressed here.

\subsection{Process-matrix tomography}

A process matrix can always be reconstructed using \emph{informationally complete} instruments. This is directly analogous to informationally complete measurements \cite{newtonmeasurability1968, bandempirical1970,
prugovecki1977}. For this purpose, note that a process matrix $W^{L_1\dots L_n}$ can be formally regarded as a density matrix with different normalisation:
\begin{equation}
\rho^{L_1\dots L_n}:= \frac{1}{d_O} W^{L_1\dots L_n},
\label{density}
\end{equation}
where $d_O=\prod_{j=1}^nd_{O_j}$ is the product of the dimensions of all output spaces. The matrix \eqref{density} is normalised as a quantum state, $\tr \rho=1$, because process matrices are normalised as $\tr W =d_O$ \cite{oreshkov12, araujo15}. It is always possible to find an informationally complete POVM with product elements $E_{x_1}^{L_1}\otimes\dots \otimes  E_{x_n}^{L_n}$, $\sum_{x_j} E_{x_j}^{L_j} = \id^{L_j}$, 
such that $\rho$ is a function of the probability distribution
\begin{multline}
P(x_1,\dots,x_n)
=\tr\left[E_{x_1}^{L_1}\otimes \dots \otimes E_{x_n}^{L_n} \cdot \rho^{L_1\dots L_n}\right].
\label{rhotomography}
\end{multline}
This means that $W$ is also a function of the same distribution.
Define now $M_{x_j} := E_{x_j}/d_{O_j}$ for $j=1,\dots,n$. As the sets ${\cal J}_j := \left\{M_{x_j}\right\}_{x_j}$ are properly normalised instruments, the statistics \eqref{rhotomography} can be obtained by implementing such instruments on the original process matrix:
\begin{multline}
P(x_1,\dots,x_n)\\
=\tr\left[M_{x_1}^{L_1}\otimes \dots \otimes M_{x_n}^{L_n} \cdot W^{L_1\dots L_n}\right].
\label{Wtomography}
\end{multline}
Thus, it is possible to reconstruct $W$ by measuring the instruments ${\cal J}_1,\dots,{\cal J}_n$. An example of this procedure is the \emph{causal tomography} considered in Ref.~\cite{ried2015quantum}.

Alternatively, one can use the fact that a process matrix defines a CPTP map from all outputs, $\bigotimes_{j=1}^nO_j$, to all inputs, $\bigotimes_{j=1}^nI_j$, and such a CPTP map can be reconstructed using conventional \emph{process tomography} \cite{chuang00}. Informationally complete \emph{testers} \cite{chiribella09b} also allow reconstructing a process matrix, although in general they require \emph{a priori} knowledge of the causal order between laboratories. In the context of causal discovery, such a knowledge is not available and thus arbitrary testers cannot be used.

Given a process with a Markovian causal explanation, it is in principle possible to gain control over all connecting mechanisms and, through tomography, reconstruct the process matrix of `wires' \eqref{connectionsonly}. This matrix contains full information about the causal structure. However, this may well be unfeasible in practice; the methods discussed below can provide information about causal structure in the absence of full control.

\subsection{Discovery using Hilbert-Schmidt basis}
Consider a space of linear operators over a Hilbert space with finite dimension $d_X$, $X\equiv \lin({\cal H}^{X})$. A Hilbert-Schmidt (HS) basis is a set of self-adjoint matrices $\left\{\sigma_{\mu}^{X_{\mu}}\right\}_{\mu=0}^{d^2-1}$ with $\sigma_0 =\id$, $\tr \sigma_{\mu}\sigma_{\nu} =d_X\delta_{\mu\nu}$. This implies $\tr \sigma_j=0$ for $j\geq 1$.\footnote{Pauli matrices are an example for $d=2$.} Local operations and process matrices can be expanded in HS bases. For example, a generic CP map for a local laboratory $A=A_I\otimes A_O$ decomposes as
\begin{equation}
M^{A_IA_O} = \sum_{\mu,\, \nu=0}^{d^2-1} v_{\mu \nu} \, \sigma^{A_I}_{\mu}\otimes \sigma^{A_O}_{\nu},\quad v_{\mu \nu} \in \mathbb{R}.
\label{decomposition}
\end{equation}
Following Ref.~\cite{oreshkov12}, the following terminology applies.
\begin{itemize}
	\item The term proportional to identity, $v_{00}\id$, is called \emph{of type 1}.
	\item Terms equal to the identity on all subsystems except subsystem $X$ are called \emph{of type $X$}. For example, the terms $v_{0 j} \id^{A_I}\otimes \sigma^{A_O}_{j}$, for $j>0$, are of type $A_O$.
	\item Terms equal to the identity on all subsystems except $X_1\otimes\dots \otimes X_k$ are called \emph{of type $X_1\cdots X_k$}. For example, the terms $v_{j l} \sigma^{A_I}_{j}\otimes \sigma^{A_O}_{l}$, for $j,l>0$, are of type $A_IA_O$.
	\item The sum of terms of different types $X_1,\dots,X_r$ is called \emph{of type $(\sum_{j=1}^r X_j)$}. For example, a generic CP map \eqref{decomposition} has terms of type $(1+A_I+A_O+A_IA_O)$.	
\end{itemize}
Some simple algebra can be applied to the above notation, such as $1\cdot X=X$ and $(\sum_jX_j)\cdot (\sum_lY_l)=(\sum_{jl}X_jY_l)$ for $X_j\not=Y_l$.

The HS decomposition can be used to characterise correlations and causal relations. For example, given a bipartite process matrix $W^{AB}$, terms of type $A_OB_I$ allow signalling from $A$ to $B$. Terms of type $A_IB_I$, on the other hand, are responsible for common cause correlations when POVM elements $E^{A_I}\otimes F^{B_I}$ are measured.

For a quantum channel $T^{A_O B_I}$, the trace-preserving condition, Eq.~\eqref{CPTP}, implies that terms of type $A_O$ vanish. 
Consequently, a process matrix that factorises over a DAG $G$, Eq.~\eqref{quantummarkov}, has to satisfy a set of constraints. It is thus possible to test the compatibility of a process matrix with a DAG $G$ by looking for the terms excluded by $G$: if an excluded term is found, the process matrix is not compatible with $G$. Note that, although the constraints are formulated with reference to a HS basis, they are in fact basis independent, as any local change of HS basis does not change the type of terms contained in a process matrix. Such conditions can also be reformulated as basis-independent linear constraints, following the methods introduced in \cite{araujo15}, but this technique will not be discussed here.

\subsection{Faithfulness}\label{faithfulness}
A general process matrix might factorise over more than one DAG. Consider the example, Fig.~\ref{finetuned}(a),
\begin{equation}
W^{A_IA_OB_IB_O}=\rho_A^{A_I}\otimes \rho_B^{B_I}\otimes \id^{A_OB_O},\quad
\label{simple}
\end{equation}
with $\tr\rho_A=\tr\rho_B=1$.
This can be factorised in the following ways (with tensor product symbols understood):
\begin{align}
W =& \rho_A^{A_I} T^{A_OB_I} \id^{B_O},\quad T^{A_OB_I} = \id^{A_O} \rho_B^{B_I},\\
W =& \rho_B^{B_I} T^{B_OA_I}\id^{A_O}, \quad T^{B_OA_I}= \id^{B_O}\rho_A^{A_I},\\
W =& T_A^{A_I}T_B^{B_I}\id^{A_OB_O}, \quad T_A^{A_I}=\rho_A^{A_I},\; T_B^{B_I}=\rho_B^{B_I},
\end{align}
which correspond to the DAGs in Figs.~\ref{finetuned}(b), \ref{finetuned}(c), and ~\ref{finetuned}(d), respectively.%
\begin{figure}%
\begin{tabular}{m{12em}l}
(a) & (b) \vspace{0.5em}\\
\includegraphics[width=0.35\columnwidth]{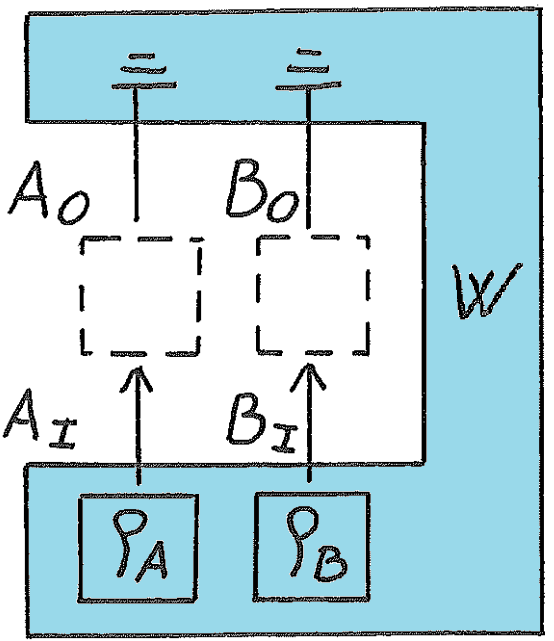} &
\begin{tabular}{l}
\includegraphics[width=0.35\columnwidth]{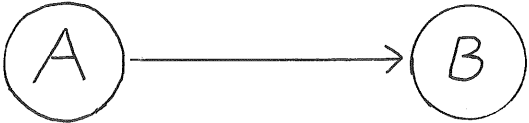} \vspace{1em} \\
(c) \vspace{0.5em} \\
\includegraphics[width=0.35\columnwidth]{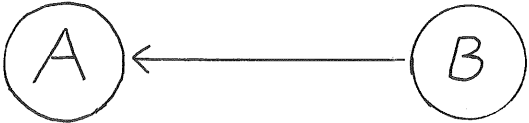} \vspace{1em} \\ 
(d) \vspace{0.5em} \\
\includegraphics[width=0.35\columnwidth]{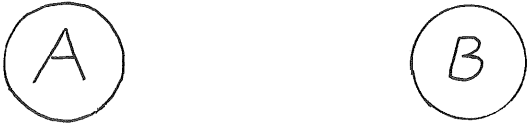}
\end{tabular}
\end{tabular}
\caption{\textbf{Faithfulness.} The process (a), Eq.~\eqref{simple}, satisfies the quantum causal Markov condition for all three depicted causal structures. However, only (d) corresponds to a \emph{faithful} causal model.}%
\label{finetuned}%
\end{figure}
However, it is clear that only the third DAG provides an appropriate representation of the causal relations in $W$, as it describes two uncorrelated laboratories.
The DAGs in Figs.~\ref{finetuned}(c), \ref{finetuned}(d), on the other hand, require a specially selected mechanism that is unable to carry causal influence. This leads to the following definition (analogous to the definition for classical causal models \cite{spirtes2000causation}):
\begin{definition}[Faithfulness]
An MQCM $\left\langle G,W \right\rangle$ is \emph{\textbf{faithful}} if $W$ contains non-vanishing HS terms of all types compatible with the factorisation \eqref{quantummarkov}. Such terms are said to be \emph{implied} by $G$.
\end{definition}
The terms implied by a DAG can be readily identified. To each laboratory $L_j$ is associated a quantum channel connecting the parent space $\ps_j$ to the input space $I_j$. Such a channel is given by a matrix that satisfies Eq.\ \eqref{modelparameters}, $\tr_{I_j}T_j^{\ps_j I_j} = \id^{\ps_j}$. This condition implies that $T_j^{\ps_j I_j}$ cannot contain terms of type $S_e$, for $e\in \IN_j$, or any product thereof. The remaining terms are of type $1+ \prod_{e\in \IN_j}(1+S_e) I_j$. Thus, a DAG $G=\left\langle {\cal L},{\cal E} \right\rangle$, \lab, implies all and only terms of type
\begin{equation}
\prod_{j=1}^n \left[1+ \prod_{e\in \IN_j}(1+S_e) I_j\right].
\label{faithful}
\end{equation}
In particular, every edge $e$ ending on a laboratory $L_j$ implies terms of type $S_eI_j$. Recall, for an edge $e$ connecting a laboratory $L_k$ to a laboratory $L_j$, the source space $S_e$ is a subsystem of the output space $O_k$. At the price of a slight abuse of terminology, we will refer to terms of type $S_eI_j$, with $e\in\out_k\cap\IN_j$, simply as terms of type $O_kI_j$. We can thus say that a term $O_kI_j$ is implied if and only if $L_k$ is a parent of $L_j$ in the DAG.

If $\left\langle G,W \right\rangle$ is a faithful MQCM, $W$ is said to be faithful to $G$. A non-faithful causal model is said to be \emph{fine tuned}. A process matrix $W$ is fine tuned if it only factorises for fine tuned models, while it is faithful if it is faithful to some DAG.

For faithful causal models, the relations between laboratories in a DAG directly correspond to causal relations. In Appendix \ref{faithfulproof}, the following theorem is proved.
\begin{theorem}
\label{faithfultheorem}
Given a faithful causal model $\left\langle G, W\right\rangle$ for a set of laboratories \lab, a laboratory $L_k$ is a \emph{direct cause} for a laboratory $L_h$ if and only if $L_k$ is a parent of $L_h$ in $G$.
\end{theorem}

The `only if' part of the theorem is also true for non-faithful models. This property can be phrased as \emph{children are able to screen off the causal influence of their parents to the rest of the world}. This means that there exist instruments for the child laboratories that break the causal connections from their common parents to any other laboratory.

\subsection{Discovery of faithful causal structures}\label{faithfuldiscovery}
The task of causal discovery is to find a causal structure, i.e.\ a DAG $G$, which can explain a set of empirical data. Faithful quantum causal models have two important properties in this respect:
\begin{inparaenum}[i)]
	\item \label{unambiguous} the causal structure of a faithful \qdag can be discovered unambiguously; 
	\item almost all causal models are faithful, in an appropriate measure-theoretic sense.
\end{inparaenum}
These points are formalised in the following theorems.
\begin{theorem}
If two MQCMs $\left\langle G, W\right\rangle$, $\left\langle G', W\right\rangle$ for the same process matrix $W$ are both faithful, then $G=G'$.
\end{theorem}
\begin{proof}
As $G$ and $G'$ have the same set of vertices, the only thing to prove is that they have the same set of edges. This is equivalent to proving that each laboratory has the same set of parents for $G$ and $G'$. The latter is a direct consequence of the fact that $W$ contains a term of type $O_kI_j$ if and only if $L_k$ is a parent of $L_j$, Eq.~\eqref{faithful}. 
\end{proof}
Thus, the HS terms contained in a faithful process matrix uniquely identify a causal structure. The HS terms can be fully determined if in each laboratory an informationally complete instrument is available. If only a subset of the HS terms is measured, the observed data can be compatible with more than one faithful causal model. This is analogous to the classical case: if arbitrary interventions are not available, it is not possible to fully characterise causal structure.

The next theorem shows that, given a sufficiently regular prior knowledge on the causal structure of a model, there is vanishing probability that the model is fine-tuned (see Ref.~\cite{Meek95} for the classical analogue).
\begin{theorem}
Given a set of laboratories ${\cal L}=\left\{L_j=I_j\otimes O_j\right\}_{j=1}^n$, a DAG $G=\left\langle {\cal L}, {\cal E}\right\rangle$ and a non-singular measure on the space of process matrices that factorise over $G$, the set of fine-tuned process matrices have measure zero.
\end{theorem}
\begin{proof}
The set of process matrices that factorise over $G$ can be parametrised by $m$ real parameters, where $m$ is the number of HS terms implied by $G$ (excluding the term of type $1$, which has a constant coefficient).  A non-singular measure on this set is a measure on $\mathbb{R}^m$ that is non-singular with respect to the Lebesgue measure. The set of fine-tuned process matrices is defined by having exactly vanishing coefficients for some of the implied HS terms, which is a zero-measure set for every non-singular measure.
\end{proof}
Note that there is only a finite number of DAGs with the same vertex set. Thus, a nonsingular measure over the space of causal models for a given set of laboratories decomposes as a finite sum of nonsingular measures for the individual DAGs. This implies the following corollary:
\begin{corollary}
Given a non-singular measure on the space of MQCMs for a given set of laboratories ${\cal L}=\left\{L_j=I_j\otimes O_j\right\}_{j=1}^n$, the set fine-tuned causal models has measure zero.
\end{corollary}
The above results indicate that, given a set of possible causal explanations, a faithful one should be preferred because, unless additional information is available, a vanishing probability is assigned to fine-tuned models.

In the presence of latent laboratories, faithfulness might not be sufficient to single out a unique causal model. In this case, further assumptions would be needed to decide among a set of faithful causal explanations.

\section{Classical limit}\label{limit}
It is important that classical causal models can be recovered from quantum ones in an appropriate limit. To ensure we are recovering classical models that can be considered \emph{causal}, it is necessary to show that these models support identification via classical interventions. Recall, this is the hallmark of the classical causal modelling methodology. As it will be shown in this section, to recover classical causal models it is sufficient that all available local operations can be described as classical interventions. The particular case of classical models conditioned on `idle' interventions, or equivalently `purely observational' Bayesian networks, is recovered by further restricting local operations to projective measurements in a fixed basis\blue{\footnote{Indeed, classical non-invasive measurements correspond, in the quantum formalism, to projective measurements in a fixed basis: as repeating such measurements confirms previously obtained outcomes, they can be interpreted as revealing pre-existing properties of the system without disturbing it.}.}
It is crucial that \emph{all} classical causal structures can be recovered in this way. This was not possible, for example, in the framework of Ref.~\cite{pienaar2014graph}, where only a subset of DAGs was recovered in the classical limit.

Quantum systems effectively behave as classical ones in the limit where it is only possible to access and manipulate states in a fixed basis that factorises over separated systems (called \emph{pointer basis}). 
Whether this condition is enforced by decoherence \cite{ZurekPointer81}, collapse models \cite{Ghirardimodel85}, `fuzzy measurements' \cite{PeresEmergence92}, or in other ways, will not be discussed here. Rather, it will be shown that, provided the transition to classicality takes place, Markov quantum causal models reduce to classical ones.

For a local laboratory $L_j=I_j\otimes O_j$ in an MQCM, with output space factorised according to the outgoing edges, $O_j=\bigotimes_{e\in\out_j}S_e$, consider CP maps of the form
\begin{equation}
\overbar{M}_{x_j|i_j}^{I_jO_j}= \sum_{z_j \vec{o\,}_j} P(\vec{o}_j, x_j|z_j, i_j)\,\proj{z_j}^{I_j}\otimes\proj{\vec{o}_j}^{O_j},
\label{diagonal}
\end{equation}
where $x_j$ label the possible measurement outcomes, $i_j$ label the possible choices of instruments, $\vec{o}_j:=\left\{s_e\right\}_{e\in\out_j}$ and $\proj{\vec{o}_j}^{O_j} := \bigotimes_{e\in\out_j}\proj{s_e}^{\opa_e}$. For a fixed basis, CP maps of this form can be interpreted as follows: a classical system enters the laboratory in some state labelled by $z_j$, a value $x_j$ is recorded as a measurement outcome and, for each outgoing edge $e$, a system is sent out of the laboratory in the state $s_e$. The parameter $i_j$ determines the conditional probability for $x_j$ and $\vec{o}_j$ given $z_j$.

In this interpretation, $x_j$ are the observed variables and $i_j$ are intervention variables. $z_j$ and $\vec{o}_j$, on the other hand, are latent variables, as they are not observed directly. In order to recover a causal model for the variables $x_j$, without such additional latent variables, it is necessary to further restrict the possible local operations. In particular, the variables $\vec{o}_j$, responsible for transmitting causal influence out of the laboratory, should be influenced directly only by the observed variables $x_j$ (and only indirectly by $i_j$ and $z_j$):
\begin{equation}
P(\vec{o}_j, x_j|z_j, i_j) = P(\vec{o}_j| x_j) P(x_j|z_j, i_j).
\label{anotherprobability}
\end{equation}

Given the above conditions, the statistics generated by a Markov quantum causal model are equivalent to those generated by a classical causal model \emph{with the same causal structure}. The proof of the following theorem is in Appendix \ref{prooflimit}.
\begin{theorem}
\label{limittheorem}
Given an MQCM $\left\langle G,W\right\rangle$ for a set of laboratories \lab, the conditional probability
\begin{multline}
P(x_1,\dots,x_n|i_1,\dots,i_n)  \\
= 
\tr\left[\bigotimes_{j=1}^n \overbar{M}_{x_j|i_j}^{L_j}\cdot W^{L_1\dots L_n}\right]
\label{diagonalprob}
\end{multline}
obtained from CP maps of the form \eqref{diagonal} with probabilities of the form \eqref{anotherprobability},
satisfies the causal Markov condition \eqref{causalmarkov} for a DAG $G_c$ isomorphic to $G$.  $G_c$  has random variables $X_1,\dots,X_n$ as vertices, where $X_j$ takes as values the labels $x_j$ of the local measurement outcomes.
\end{theorem}

It is also possible to show that each classical causal model can be reproduced as the classical limit of an MQCM with the same causal structure, i.e.\ with an isomorphic DAG. Replace each random variable $X_j$ with a local laboratory $L_j=I_j\otimes O_j$, where $I_j$ has a basis element $\ket{x_j}$ for each value of $X_j$ and $O_j$ is the tensor product of a copy of $I_j$ for each outgoing edge. The process matrix is then 
\begin{equation}
\begin{split}
W =&\bigotimes_{j=1}^n T_j^{\ps_j I_j},\\
T_j^{\ps_j I_j} =&  \sum_{z_j \psv_j} P(z_j | \psv_j) \proj{\psv_j}^{\ps_j}\otimes\proj{z_j}^{I_j}, 
\end{split}
\label{classicalW}
\end{equation}
with coefficients $P(z_j|\psv_j)=\pr(X_j=z_j|P\!A_j=\psv_j, I_j=\textrm{idle})$ given by the classical model. The `observed' probability distribution (where no interventions are made, from the perspective of the classical model) is obtained when all instruments are restricted to be projective measurements in the pointer basis:
\begin{equation}
M^{L_j}_{x_j} = \proj{x_j}^{I_j}\bigotimes_{e\in\out_j}\proj{x_j}^{\opa_e}.
\label{natural}
\end{equation}

Interventions in the classical model correspond to more general diagonal operations of the form \eqref{diagonal}.
In particular, `do' interventions are realised by ignoring the input and preparing the chosen value on the output space: $M^{L_j}_{\textrm{do}(x_j)} = \id^{I_j}\bigotimes_{e\in\out_j}\proj{\textrm{do}(x_j)}^{\opa_e}$. The split of classical nodes into input and output reproduces  the \emph{single node intervention graphs} (SWIGs), studied in the context of classical causal models \cite{richardsonsingle2013}.

The result above not only shows that classical causal models can be subsumed by quantum ones, but also provides a direct way to apply classical tools to quantum models. Indeed, classical models are recovered under restrictions on the performed operations. Such restrictions can arise because of uncontrolled factors, such as environmental decoherence, but it is also possible to simply \emph{choose} to use only instruments compatible with a classical description. Thus, all properties of classical causal models directly extend to quantum ones, conditioned on the choice of classical instruments. For example, in classical causal models, if $X_3$ is the only common cause of $X_1$ and $X_2$, it follows from the causal Markov condition that $X_1$ and $X_2$ are conditionally independent given $X_3$: $P(X_1,X_2|X_3) = P(X_1|X_3)P(X_2|X_3)$, where it is understood that such a property holds conditioned on the `idle' intervention. Such a property still holds for quantum causal models, conditioned on the choice of appropriate instruments: If $L_3$ is a quantum common cause of $L_1$ and $L_2$, then $P(M^{L_1}_1,M^{L_2}_2|M^{L_3}_3,{\cal J}^{L_1}_1,{\cal J}^{L_2}_2,{\cal J}^{L_3}_3)=P(M^{L_1}_1|M^{L_3}_3,{\cal J}^{L_1}_1,{\cal J}^{L_3}_3)P(M^{L_2}_2|M^{L_3}_3,{\cal J}^{L_2}_2,{\cal J}^{L_3}_3)$ for all instruments ${\cal J}^{L_1}_1,\dots$ satisfying conditions \eqref{diagonal}, \eqref{anotherprobability}.

\section{Beyond definite and acyclic causal structures}\label{beyond}
The results of the previous section show that the causal structure of a Markov process, represented by a DAG, remains unaltered in the quantum-to-classical transition. Thus, the quantum causal models discussed above are `quantum' only by virtue of the physics describing the systems that convey causal influence. The \emph{causal structure} itself is not different to that of classical models. 

The formalism, however, can be naturally extended to more general structures. \emph{Probabilistic mixtures} of causal structures are a natural extension. Owing to the linearity of the generalised Born rule, Eq.~\eqref{born}, these are represented as
\begin{equation}
W=\sum_G p(G) W^{(G)},\quad \sum_{G}P(G)=1, \quad P(G)\geq 0,
\label{mixture}
\end{equation}
where the sum is over DAGs $G$ and $W^{(G)}$ is a process matrix that factorises over $G$.

A process of the form \eqref{mixture} is still compatible with a definite, albeit unknown, causal structure, which can be understood as an objective property that exists independently of the operations performed. The natural next step is to look for causal structures that might be, in some sense, \emph{indefinite}. That space-time itself might have quantum properties in a theory of quantum gravity motivates this idea \cite{hardy2007towards, zych14}.

An example is the `quantisation' of the mixture \eqref{mixture}. Given a set of laboratories \lab, introduce two additional laboratories, $C\equiv C_O$ (with trivial input) and $D\equiv D_I$ (with trivial output). Let $W^{(G)}=\proj{w^{(G)}}$ be a rank-one projector which factorises over the DAG $G$ (this is the case if all mechanisms are unitary matrices or pure-state preparations). A \emph{quantum-controlled} causal structure is then given by the process matrix
\begin{equation}
W=\proj{w}, \quad \ket{w} = \sum_G \ket{G}^{C_O}\otimes \ket{G}^{D_I}\otimes \ket{w^{(G)}},
\label{controldag}
\end{equation}
where the basis vectors $\ket{G}$ range over the set of controlled DAGs. If $C$ prepares a state $\ket{G}^{C_O}$, the laboratories in ${\cal L}$ find themselves in the causal relations dictated by $G$ (and the control system is transferred undisturbed to $D$). However, if $C$ prepares an arbitrary superposition $\sum_G \psi_G \ket{G}^{C_O}$, and $D$ also measures in a superposition basis, one can expect causal relations that are incompatible with any DAG\footnote{This should not be expected, however, when the superposition only involves DAGs with the same partial order between laboratories, as this results in a causally ordered process matrix, which has a Markovian causal explanation, see Sec.~\ref{latent}.}. Indeed, an example of quantum-controlled causal structure is the \emph{quantum switch} \cite{chiribella09, colnaghi11}, where the above sum ranges over DAGs of type $L_{\sigma(1)}\rightarrow L_{\sigma(2)}\dots \rightarrow L_{\sigma(n)}$, with $\sigma$ a permutation of $n$ elements. This type of resource can be shown to be incompatible with any definite causal order \cite{araujo15, oreshkov15} and thus in particular with any DAG. Additionally, the quantum switch can provide an advantage for several tasks \cite{chiribella12, araujo14, feixquantum2015} and an experimental proof-of-principle was recently demonstrated \cite{procopio_experimental_2014}, showing the potential practical relevance of quantum causal structures. 

The process \eqref{controldag} can be interpreted as the superposition of different amplitudes, each corresponding to a directed, acyclic causal structure. It is an interesting question whether more general causal models can also be understood in a similar way. An indication that this might not be the case is the fact that processes of the form \eqref{controldag} cannot violate \emph{causal inequalities} \cite{araujo15, oreshkov15} whereas more general processes allowed by the formalism can \cite{oreshkov12, branciard15}. Intriguingly, this is also true for classical systems: there exist causal models, locally compatible with classical physics, that are incompatible with any causal order or mixture of causal orders \cite{baumeler14, baumelerspace2015}. 

The possibility of closed time-like curves (CTCs) in general relativity \cite{morris1988wormholes} motivates considering models that contain directed cycles. The process matrix formalism provides a natural framework for studying CTCs as cyclic causal structures, as the possibility of interventions, and thus the causal interpretation of the model, does not rely on the acyclicity property. It is not clear, however, whether indefinite and acyclic causal structures are in fact separate concepts. Furthermore, it is unclear whether and how core notions such as the causal Markov condition and faithfulness should be generalised beyond definite, acyclic causal structures.

\section{Conclusions}
This work has foundational implications insofar as it shows that quantum mechanics has a causal interpretation in a similar manner to classical mechanics. Cause-effect relations are identified with correlations between controlled and observed events, and a causal structure is a set of transformations that define a DAG.

The findings presented here do not reproduce the wealth of results from the literature on classical causal models. However, the surprising similarity between the frameworks suggests that further techniques from classical causal modelling may be applicable to the quantum case. ‘Quantising’ causal models ought to be a rich and promising field of research.

Causal discovery plays a prominent role in classical machine learning and ought to play a similarly pivotal role in the the emerging field of quantum machine learning \cite{Schuldintroduction15}. The mathematical formalism for quantum causal discovery introduced here can provide a foundation for further development in this direction.

\begin{acknowledgments}
We are grateful to Peter Evans, Gerard Milburn, and Magdalena Zych for discussions and feedback.
This work was supported in part by the Templeton World Charity Foundation (TWCF 0064/AB38).
\end{acknowledgments}

\appendix

\section{Proof of theorem \ref{faithfultheorem}}
\label{faithfulproof}
If $L_k$ is a parent of $L_h$ then, by virtue of faithfulness, $W$ contains terms of type $O_kI_h$. Given an arbitrary set of CPTP maps $C_j^{L_j}$, $j\not=k,h$, their tensor product contains the term of type $1$, which implies that the reduced process
\begin{equation*}
\overline{W}^{L_kL_h} =\tr_{L_j\not=L_k, L_h}\left[\bigotimes_{j\not=k,h}C_j^{L_j}\cdot W\right]
\end{equation*}
also contains a term of type $O_kI_h$. As such a term can be used for signalling, $L_k$ has causal influence over $L_h$
for any instruments in the remaining laboratories, thus $L_k$ is a direct cause for $L_h$.
Conversely, assume that $L_k$ is not a parent of $L_h$. Then, every HS term nontrivial on  $O_k$ must also be nontrivial on an input $I_j$ such that $L_k$ is a parent of $L_j$ (and thus $j\not= h$). Let $\sigma^{O_kI_jX}$ be the corresponding element of the HS basis (where $X$ is any additional subsystem on which $\sigma$ is non-trivial); then, for the CPTP map $C^{L_j}=\id^{L_j}/{d_{O_j}}$, and recalling that HS basis elements have null partial trace for any subsystem in which they are not trivial,
\begin{equation*}
\tr_{L_j}C^{L_j} \sigma^{O_kI_jX} \equiv \tr_{L_j} \sigma^{O_kI_jX} /{d_{O_j}}=0,
\end{equation*}
which means that the reduced process matrix obtained when $L_j$ performs the maximally noisy channel does not contain the HS term $O_kI_jX$. If the maximally noisy channel is performed in every laboratory that has $L_k$ as a parent, then the resulting reduced process does not contain any term of the form $O_kX$, which means that no signalling from $L_k$ is possible in the reduced model. Thus, $L_k$ is not a direct cause of any laboratory of which it is not a parent.

\section{Proof of theorem \ref{limittheorem}}
\label{prooflimit}
The probabilities \eqref{diagonalprob} are equivalently generated by the diagonal process matrix
\begin{equation}
\begin{split}
\overline{W}= \sum_{z_1 \vec{o\,}_1}\cdots \sum_{z_n \vec{o\,}_n} \tr\left(\proj{\vec{v}}\cdot W\right) \proj{\vec{v}}, \\
\proj{\vec{v}}= \bigotimes_{j=1}^n 
 \proj{z_j}^{I_j}\otimes\proj{\vec{o}_j}^{O_j}.
\end{split}
\label{Wdiagonal}
\end{equation}
Using the quantum causal Markov condition \eqref{quantummarkov}, the above matrix can be written as a product of CPTP maps diagonal in the pointer basis:
\begin{equation}
\begin{split}
\overline{W} =&\bigotimes_{j=1}^n \overline{T}_j^{\ps_j I_j},\\
\overline{T}_j^{\ps_j I_j}:=&  \sum_{z_j \psv_j} P(z_j | \psv_j) \proj{\psv_j}^{\ps_j}\otimes\proj{z_j}^{I_j}, \\
P(z_j | \psv_j) :=&\tr\left(\proj{\psv_j}^{\ps_j}\otimes\proj{z_j}^{I_j}\cdot T_j^{\ps_j I_j}\right),
\end{split}
\label{diagonalW}
\end{equation}
with $\psv_j :=\left\{s_e\right\}_{e\in \IN_j}$, $\proj{\psv_j}^{\ps_j}:= \bigotimes_{e\in\IN_j}\proj{s_e}^{\opa_e}$.
For diagonal operations, the generalised Born rule reduces to the composition of probabilities:
\begin{multline}
P(\textbf{x}|\textbf{i}) 
= \tr \left[\bigotimes_j \overbar{M}^{L_j}_{x_j|i_j}\cdot W\right] 
\\ 
= \tr \left[\bigotimes_j \overbar{M}^{L_j}_{x_j|i_j}\cdot\overline{W}\right] 
\\
= \sum_{\textbf{z}\, \textbf{o}} P_{L}(\textbf{o}, \textbf{x}|\textbf{z}, \textbf{i})P_{W}(\textbf{z}|\textbf{o}),
\label{classicalcomposition}
\end{multline} 
where $\textbf{o} := \left\{\vec{o}_j\right\}_{j=1}^n \equiv \left\{\psv_j\right\}_{j=1}^n$, 
$\textbf{x} := \left\{x_j\right\}_{j=1}^n$, $\textbf{i} := \left\{i_j\right\}_{j=1}^n$, $\textbf{z} := \left\{z_j\right\}_{j=1}^n$,
and
\begin{equation}
\begin{split}
P_{L}(\textbf{o}, \textbf{x}|\textbf{z}, \textbf{i}) :=& \prod_{j=1}^n P(\vec{o}_j, x_j|z_j, i_j),\\
P_{W}(\textbf{z}|\textbf{o}) :=& \prod_{j=1}^n P(z_j | \psv_j)
\end{split}
\label{details}
\end{equation}
are the conditional probabilities corresponding to the coefficients of local operations \eqref{diagonal} and mechanisms \eqref{diagonalW}, respectively.
Combining the expressions above, and using condition \eqref{anotherprobability}, gives 
\begin{multline}
P(x_1,\dots,x_n|i_1,\dots,i_n)  \\
= \sum_{\textbf{o}} \prod_{j=1}^n \left[\sum_{z_j} P(\vec{o}_j, x_j|z_j, i_j) P(z_j | \psv_j)\right] \\
= \sum_{\textbf{o}} \prod_{j=1}^n P(\vec{o}_j| x_j)P(x_j|\psv_j, i_j) \\
= \sum_{s_e,\,e\in{\cal E}} \prod_{j=1}^n P(\left\{s_e\right\}_{e\in \out_j}| x_j)P(x_j|\left\{s_e\right\}_{e\in \IN_j}, i_j) \\
= \prod_{j=1}^n P(x_j|pa_j, i_j),
\label{quantumtoclassical}
\end{multline}
where $pa_j:=\left\{x_k|(L_k, L_j)\in {\cal E}\right\}$ is the set of outcomes associated with the parents of node $j$ with respect to the original graph $G$.

\small


\raggedleft



\end{document}